\documentclass[11pt]{article}
\usepackage{amsmath,amssymb,amsthm,amsfonts,hhline,color}
\usepackage{etoolbox, textcomp}
\patchcmd{\thebibliography}{\section*{\refname}}{}{}{}

\def\Z{{\mathbb Z}}

\title{W-algebras for Argyres--Douglas theories }

\author{Thomas Creutzig\thanks{Department of Mathematical and Statistical Sciences, University of Alberta,
Edmonton, Alberta  T6G 2G1, Canada. 
email: creutzig@ualberta.ca}} 

\date{}

\begin{document}

\bibliographystyle{amsalpha}

\theoremstyle{plain}
\newtheorem*{introthm}{Theorem}
\newtheorem{obs}{Observation}
\newtheorem{thm}{Theorem}[section]
\newtheorem{prop}[thm]{Proposition}
\newtheorem{lem}[thm]{Lemma}
\newtheorem{cor}[thm]{Corollary}
\newtheorem{conj}[thm]{Conjecture}

\theoremstyle{definition}
\newtheorem{defi}[thm]{Definition}
\newtheorem{rem}[thm]{Remark}

\newcommand {\AL}{AL}
\newcommand {\CC}{\mathbb{C}}
\newcommand {\ZZ}{\mathbb{Z}}
\newcommand {\tr}{\text{tr}}
\newcommand {\ch}{\text{ch}}
\newcommand {\sch}{\text{sch}}
\newcommand {\sltwo}{\mathfrak{sl}_2}
\newcommand {\cW}{\mathcal{W}}
\newcommand {\cX}{\mathcal{X}}
\newcommand {\cB}{\mathcal{B}}
\newcommand {\cS}{\mathcal{S}}
\newcommand {\cO}{\mathcal{O}_p}
\newcommand {\voa}{vertex operator algebra}
\newcommand {\voas}{vertex operator algebras}
\newcommand{\Sing}{M(p)}
\newcommand{\Trip}{W(p)}

\newcommand{\hopflink}{{\text{\textmarried}}}

\renewcommand{\baselinestretch}{1.2}

\maketitle

\begin{abstract}
The Schur-index of the $(A_1, X_n)$-Argyres--Douglas theory is conjecturally a character of a \voa. Here such vertex algebras are found for the $A_{\text{odd}}$ and $D_{\text{even}}$-type  Argyres--Douglas theories. The \voa{} corresponding to $A_{2p-3}$-Argyres-Douglas theory is the logarithmic $\mathcal B_p$-algebra of \cite{CRW}, while the one corresponding to $D_{2p}$, denoted by $\cW_p$, is realized as a non-regular Quantum Hamiltonian reduction of $L_{k}(\mathfrak{sl}_{p+1})$ at level $k=-(p^2-1)/p$. For all $n$ one observes that the quantum Hamiltonian reduction of the \voa{} of $D_n$ Argyres--Douglas theory is the \voa{} of $A_{n-3}$ Argyres--Douglas theory.  As corollary, one realizes the singlet and triplet algebras (the vertex algebras associated to the best understood logarithmic conformal field theories) as Quantum Hamiltonian reductions as well.
Finally, characters of certain modules of these \voas{} and the modular properties of their meromorphic continuations are given.    

\end{abstract}

\newpage

\section{Introduction}

Vertex operator algebras are a mathematicians attempt to formalize the notion of the symmetry algebra of two dimensional conformal field theory. Recently a close connection between supersymmetric four-dimensional conformal field theories and \voas{} has been observed \cite{GRRY}. In physics terminology a \voa{} is often named a chiral algebra. Two examples are that certain protected states of the four-dimensional $\mathcal N=2$ super conformal field theory are associated to a chiral algebra \cite{BLLPRR}; and that one associates \voas{} to two-dimensional surfaces with boundary. Gluing of surfaces then corresponds to maps between \voas{} as e.g. cosets or quantum Hamiltonian reductions, see \cite{T}.

One instant of this first phenomenon are Argyres-Douglas theories \cite{AD} whose Schur-index is believed to be the character of a \voa{}, for literature on the topic see e.g. \cite{CS, CGS2, BN, BR, BPRR, BLLPRR, XYY}. The objective of this note is to find such \voas{} for the missing type $(A_1, X_n)$-Argyres--Douglas theories. These are the cases of $D_{2n}$ and $A_{2n+1}$.
 The relevant \voas{} are the $\cB_p$-algebra of \cite{CRW} and a certain non-regular Quantum Hamiltonian reduction $\cW_p$ of the simple affine \voa{} $L_{k}(\mathfrak{sl}_{p+1})$ at level $k=-(p^2-1)/p$. For details on Quantum Hamiltonian reduction see \cite{Ar1, KW3}.
 The main result is the following statement, which combines Theorems \ref{thm:BPSchur}, \ref{thm:XpSchur} and \ref{thm:main}.
\begin{introthm}${}$ Let $p$ in $\mathbb Z_{\geq 2}$.
\begin{enumerate}
\item
The Schur-index of $(A_1, A_{2p-3})$ Argyres-Douglas theory and the character of the $\cB_p$-algebra agree in the following sense
\begin{equation}\nonumber
\text{ch}[\mathcal B_p] = q^{-\frac{c_p[\cB_p]}{24}} \mathcal I_{A_1, A_{2p-3}}\left(q; z^{-1}\right)
\end{equation}
with $c_p[\cB_p]= 2-6\frac{\left(p-1\right)^2}{p}$
the central charge of the $\mathcal B_p$-algebra.  
\item 
The closed formula of the Schur-index of $(A_1, D_{2p})$ Argyres-Douglas theory as conjectured by \cite{BN} and the character of the $\cW_p$-algebra agree in the following sense
\begin{equation}\nonumber
q^{-\frac{c_p[\cW_p]}{24}}\mathcal I_{(A_1, D_{2p})}(x, z; q)  = \ch[\cW_p](x, z; q). 
\end{equation}
with $c_p[\cW_p] = 4-6p$
the central charge of the $\mathcal W_p$-algebra.
\end{enumerate}
\end{introthm}
It has been understood by Christopher Beem and Leonardo Rastelli that the \voas{} of $(A_1, A_{2p-3})$ Argyres-Douglas theories are generalized Bershadsky-Polyakov algebras and this will be mentioned in \cite{BR}.

A correspondence between seemingly unrelated concepts like \voas{} and four-dimensional super symmetric conformal quantum field theory is potentially benefitial for both areas. I will in a moment list a few patterns  concerning the \voas{} $\cW_p$ and $\cB_p$. A physics interpretation of these observations would be interesting.

A key present \voa{} problem is the understanding of the non semi-simple representation category of logarithmic \voas, see  \cite{CG} for the current picture.
I find it very interesting that categorical data of the \voa{} as for example fusion products \cite{CGS} and modular data \cite{APYY} appear in the four-dimensional setting. For this reason, I provide some modular data for modules that are obtained from the \voa{} itself via twisting by spectral flow automorphism, see Theorem \ref{thm:modular}. This modular data is associated to the meromorphic continuation of characters of modules. A priori, the character of a module is a formal power series and it might converge in a suitable domain. Here, it turns out that this is indeed the case and these characters then can be meromorphically continued to certain vector-valued Jacobi forms of indefinite index. 
In analogy to the modular story of admissble level $L_k(\sltwo)$ (see Theorem 21 of \cite{CR2}) I expect that the modular properties of these Jacobi forms capture the semi-simplification of the log-modular tensor category of the corresponding \voa.

\subsection{Quantum Hamiltonian reduction at boundary admissible level}

There are very interesting patterns emerging of the interplay of $\cB_p$ and $\cW_p$ and other \voas{} related to Argyres-Douglas theory. 
Recall that given an affine \voa{} $V_k(\mathfrak g)$ at level $k$ of some Lie algebra $\mathfrak g$ and a Lie algebra homomorphism from $\sltwo$ to $\mathfrak g$ one associates a $W$-algebra via Quantum Hamiltonian reduction \cite{KW3}. If the homomorphism is trivial, then this is just the affine \voa{} again. For the principal embedding of $\sltwo$ in $\mathfrak g$ one obtains the principal or regular $W$-algebra and other non-trivial embeddings lead to $W$-algebras that are called non-regular. 
 A level $k$ is called boundary admissible \cite{KW2} if it satisfies 
\[
k+h^\vee = \frac{h^\vee}{p}
\]
for $h^\vee$ the dual Coxeter number of $\mathfrak g$ and $p$ a positive integer co-prime to $h^\vee$. At such levels, the characters of some simple \voa{} modules allow for particularly nice meromorphic extensions to multi-variable meromorphic Jacobi forms. Nice product forms of characters also appear in the context of chiral algebras coming from $M5$ branes \cite{XYY}.
 Remark \ref{rem:bdy} tells us that
\begin{obs}
The $W$-algebras corresponding to $(A_1, X_n)$ Argyres-Douglas theories are of boundary admissible level; in the cases of $A_{2n+1}$ ($n\geq 4$) and $E_7$ this observation is at the moment only conjectural. 
\end{obs}
The case of $A_{\text{odd}}$ can be verified in a similar manner as Theorem \ref{thm:main} and that will be part of a general study of the $\cB_p$-algebras \cite{ACKR}. 

An affine \voa{} is said to be conformally embedded in a larger \voa{} $W$ if both \voas{} have the same Virasoro field. Conformal embeddings in larger affine \voas{} and minimal $W$-algebras have been of recent interest \cite{AKMP, AKMP2,  AKMP3, KFPX}. 
The next observation is Corollary \ref{cor:confemb}.
\begin{obs}\
Let $p$ in $\mathbb Z_{\geq 2}$.
The simple affine \voa{} $L_k(\mathfrak{gl}_2)$ for $k+2=\frac{1}{p}$ embeds conformally in $\cW_p$.
\end{obs}
Note, that $-2+\frac{1}{p}$ is not an admissible level of $\widehat\sltwo$. 

The \voas{} of type $(A_1, D_n)$ all contain an affine \voa{} of $\sltwo$ as sub \voa. One can thus consider the Quantum Hamiltonian reduction of these \voas. These must then be extensions of the Virasoro \voa. The next observation is Corollary \ref{cor:QH} and the beginning of Section \ref{sec:Schur}.
\begin{obs}
The \voa{} of type $(A_1, A_n)$ Argyres-Douglas theory is isomorphic to the Quantum Hamiltonian reduction of the type $(A_1, D_{n+3})$
Argyres-Douglas theory. If $n$ is even this is an isomorphism of \voas{} and if $n$ is odd this is an isomorphism of Virasoro modules. 
\end{obs}
I conjecture that also in the case $n$ odd this is a \voa{} isomorphism. The work in progress \cite{A2}, see also \cite{A}, on the $\mathcal R(p)$-algebra by Drazen Adamovi\'c should be useful here. 

This observation is very natural from the physics point of view. Firstly equation (5.20) of \cite{BN} relates the Schur-index of ype $(A_1, D_{2n})$-Argyres-Douglas theory to the one of type $(A_1, A_{2n-3})$ in an analogous way as one relates affine \voa{} characters to characters of the $W$-algebra obtained via Quantum Hamiltonian reduction.
Secondly, in \cite{APYY}, Riemann surfaces are asscoiated to the Argyres-Douglas theories and the surfaces corresponding to $(A_1, A_n)$ Argyres-Douglas theory and $(A_1, D_{n+3})$ Argyres-Douglas theory only differ by an extra regular puncture. The authors of that article find a natural cap state that closes this puncture and thus maps the $(A_1, D_{n+3})$ Argyres-Douglas theory to the $(A_1, A_{n})$ Argyres-Douglas theory.

\subsection{Singlet and triplet algebras as Quantum Hamiltonian reduction}

The triplet algebras $\Trip$, for $p$ in $\mathbb Z_{\geq 2}$ are the best understood $C_2$-cofinite but non-rational \voas{} \cite{TW, AM1, AM2}. The singlet algebra $\Sing$ is its $U(1)$-orbifold \cite{CM, AM3}. These algebras are closely related to the \voas{} of Argyres-Douglas theories. 

The $\Sing$ and $\Trip$ algebras are infinte order extensions of the Virasoro algebra Vir$_p$ at central charge $c_p=1-6(p-1)^2/p$.  $\Trip$ actually carries an action of PSL$(2, \CC)$ \cite{ALM} and as PSL$(2, \CC)\ \otimes $ Vir$_p$-module
\begin{equation}\nonumber
\Trip \cong  \bigoplus_{n=0}^\infty \rho_{2n+1} \otimes L(h_{2n+1, 1}, c_p)
\end{equation}
with $\rho_{n}$ the $n$-dimensional irreducible representation of PSL$(2, \CC)$ \cite{ALM} and $L(h_{n, 1}, c_p)$ the simple highest-weight module of the Virasoro \voa{} at central charge $c_p$ at hightest-weight $h_{n,1}=((n^2-1)p-2(n-1))/4$.
The singlet on the other hand decomposes as Virasoro module as
\begin{equation}\nonumber
\Sing \cong \bigoplus_{n=0}^\infty L(h_{2n+1, 1}, c_p).
\end{equation}
Corollary \ref{cor:QH} tells us the following. Let $k=\frac{1}{p}-2$, then the two $L_k(\sltwo)$-modules inherit each the structure of a simple vertex operator algebra from $\cW_p$
\[
\mathcal C_p := \text{Com}\left( \mathcal H, \cW_p \right) \cong \bigoplus_{m=0}^\infty  L(k, 2m)
\qquad
\text{and}
\qquad
\mathcal Y_p :=  \bigoplus_{m=0}^\infty (2m+1) L(k, 2m).
\]
Here $L(k,n)$ denotes the Weyl-module at level $k$ induced from the $n+1$ dimension irreducible representation of $\sltwo$. It is simple, see Section \ref{sec:Op}.
Proposition \ref{prop:red} follows:
As Virasoro \voa-modules, we have
\[
H_{QH}(\cW_p ) \cong \mathcal B_{p}, \qquad H_{QH}(\mathcal C_p ) \cong \mathcal M(p) \qquad \text{and} \qquad H_{QH}(\mathcal Y_p ) \cong \mathcal W(p).
\]
Here $H_{QH}$ is the so-called plus reduction functor of the Quantum Hamiltonian reduction from the affine \voa{} of $\mathfrak{sl}_2$ to the Virasoro \voa{}. It follows directly from \cite[Theorem 9.1.4]{Ar2}.
I strongly belief that this isomorphism is even a \voa-isomorphism. Moreover it is tempting to conjecture as well that $\mathcal Y_p$ contains PSL$(2, \CC)$ as subgroup of automorphisms and that as PSL$(2, \CC)\, \otimes\, L_k(\sltwo)$-module
\begin{equation}\label{eq:Yp}
\mathcal Y_p \cong  \bigoplus_{m=0}^\infty \rho_{2m+1} \otimes  L(k, 2m).
\end{equation}
Remark, that \voas{} with continuous groups of outer automorphisms are important in Davide Gaiotto's picture of dualities in $\mathcal N=2$ super conformal gauge theories \cite{Gai}. These dualities relate to the geometric Langlands correspondence and so it is worth speculating that 
the resemblance of $\mathcal Y_p$ to chiral Hecke algebras of $\sltwo$ is not a coincidence.\footnote{I thank Tomoyuki Arakawa for pointing out that $\mathcal Y_p$ is similar to chiral Hecke algebras of $\sltwo$.}

\subsection{Conjectural generalizations}

The \voas{} of this work allow for two conjectural generalizations. Firstly it seems that \cite{AKMP, AKMP2, AKMP3, KFPX} and this work together are first members of a large family of conformal embeddings of affine \voas{} into boundary admissible level $W$-algebras. 
We now define a family of $W$-algebras parameterized by two positive integers $n, m$ so let us call this algebra $\cW_{n, m}$.
Consider $\mathfrak{sl}_{n+m}$, then we can embed $\mathfrak{sl}_m\oplus \mathfrak{gl}_n$ in the obvious way. Consider the Quantum Hamiltonian reduction of $V_k(\mathfrak{sl}_{n+m})$ for the $\sltwo$ embedded as follows, embed $\sltwo$ principally  in the $\mathfrak{sl}_m$ sub algebra and map it trivially in $\mathfrak{gl}_n$ so that the affine sub \voa{} is $V_{k+m-1}(\mathfrak{gl}_n)$. 
The central charge of this $W$-algebra can be computed using equation (2.3) of \cite{KW3} and one gets 
\[
c_{n, m, k} = \frac{k((n+m)^2-1)}{k+n+m} -(m-1)\left(km(m+1)+m(m^2-m-1)+n(m^2-2m-1)\right).
\]
$\cW_{n, m}$ is defined to be this $W$-algebra at level $k$, such that
\[
k+n+m = \frac{n+m}{m+1}
\]
so that this is a boundary admissible level provided $n+m$ and $m+1$ are co-prime. The central charge then becomes $-n(mn+1+m)$ which coincides with the central charge of $V_{k+m-1}(\mathfrak{gl}_n)$.
I thus conjecture
\begin{conj}
Let $n, m$ be positive integers, $n\geq 2$. Then the simple affine \voa{} $L_{-\frac{nm+1}{m+1}}(\mathfrak{gl}_n)$ embeds conformally in $\cW_{n, m}$. 
 \end{conj}
 This conjecture for $m=1$ has been proven in \cite[Theorem 5.1]{AKMP2}, for $m=2$ in \cite[Theorem 1.1]{AKMP} and we treated the case of $n=2$.
 
Generalizations of the singlet and triplet algebras exist as well \cite{FT, CM2} and have been called narrow $W$-algebras $W^0(p)_Q$. Here $Q$ denotes the root lattice of a simply laced Lie algbera and $p$ is a positive integer greater or equal to two. 
 Consider the regular quantum Hamiltonian reduction of $\cW_{n,  m}$. By this we mean the principal embedding of $\sltwo$ in $\mathfrak{sl}_n$ and the corresponding Quantum Hamiltonian reduction of $\cW_{n, m}$. 
\begin{conj}Let $m+1=p(n-1)$.
 Then the Heisenberg coset of the regular Quantum Hamiltonian reduction $H_{\text{reg}}\left(\cW_{n, m} \right)$ of $\cW_{n, m}$ is isomorphic to the narrow $W$-algebra of type $\mathfrak{sl}_n$ and parameter $p$:
 \[
 \text{Com}\left(\mathcal H, H_{\text{reg}}\left(\cW_{n, m} \right)\right) \cong W^0(p)_{A_{n-1}}.
 \]
  \end{conj}
    In this work, this conjecture has been proven in the case of $n=2$, but only as isomorphism of Virasoro modules.

\subsection{Organization}

In Section \ref{sec:Schur} the Schur-indices of Argyres-Douglas theories are quickly reviewed. 

Section \ref{sec:Xp} starts with a discussion of properties of Weyl modules for $L_k(\sltwo)$ when $k+2=\frac{1}{p}$ and $p$ in $\ZZ_{\geq 2}$. Next the $L_k(\sltwo)\otimes \mathcal H$-module $\cX_p$ is introduced. Here $\mathcal H$ denotes the rank one Heisenberg \voa. The reason for introducing $\cX_p$ is that its character coincides with the Schur index of type $(A_1, D_{2p})$-Argyres-Douglas theory. I then discuss modular and Jacobi properties of characters of twisted versions of $\cX_p$ using Appell-Lerch sums. 

In Section \ref{sec:Bp} properties of the $\cB_p$-algebra are discussed and especially it is observed that its character coincides with the 
Schur index of type $(A_1, A_{2p-3})$-Argyres-Douglas theory.

Finally in Section \ref{sec:Walgebras} the properties of $\cW_p$ mentioned in the introduction are derived; especially its character coincides with the Schur index of type $(A_1, D_{2p})$-Argyres-Douglas theory.

{\bf Acknowledgements} Most of all, I am very grateful to Shu-Heng Shao for many discussions on this topic. I also would like to thank Tomoyuki Arakawa, Chris Beem, Leonardo Rastelli and Davide Gaiotto for explanations on the relation of \voas{} and four-dimensional quantum field theory. I thank Drazen Adamovi\'c for his useful comments on the draft and discussion. Finally I very much appreciate discussions with Wenbin Yan and Ke Ye and sharing \cite{APYY} with me. 
I am supported by NSERC RES0020460 and have benefitted from the workshop "exact operator algebras in superconformal field theories" on the topic at Perimeter Institute.

\section{Schur-indices of $(A_1, X_n)$ Argyres-Douglas theories}\label{sec:Schur}

I collect some data from \cite{BN, CS}. First recall the known cases. 

The Schur index of $(A_1, A_{2n})$ Argyres-Douglas theories is identified with the vacuum character 
of the simple and rational $(2, 2n+3)$ Virasoro \voa{} at central charge $c=1-6(2n+1)^2/(4n+6)$ \cite{Ras}.
The Schur index of $(A_1, D_{2n+1})$ Argyres-Douglas theories is identified with the vacuum character 
of the simple affine \voa{} of $\sltwo$ at level $k=-4n/(2n+1)$, denoted by $L_k(\sltwo)$. Note, that $k+2=2/(2n+1)$ and hence the quantum Hamiltonian reduction of the $(A_1, D_{2n+1})$ Argyres-Douglas theory \voa{} is the \voa{} of the $(A_1, A_{2n-2})$ Argyres-Douglas theory. 
The Schur index of the $(A_1, E_n)$-Argyres-Douglas theory coincides with the vacuum character of the simple regular $W$-algebra of $\mathfrak{sl}_3$ at level $k+3=3/7$ for $E_6$ and at level $k+3=3/8$ for $E_8$. In the case of $E_7$ it seems to correspond to the sub-regular $W$-algebra of $\mathfrak{sl}_3$ at level $k+3=3/5$ (Shu-Heng Shao has verified this for the leading terms). 

\begin{rem}\label{rem:bdy}
We will see that all $W$-algebras corresponding to $(A_1, X_n)$ Argyres-Douglas theories are of boundary admissible level (in the case of $A_{\text{odd}}$ and $E_7$ this is conjectural). A level $k$ associated to a simply-laced Lie algebra $\mathfrak g$ is called boundary admissible if it satisfies $k+h^\vee=\frac{h^\vee}{p}$ for $h^\vee$ the dual Coxeter number of $\mathfrak g$  and $p$ a positive integer co-prime to the dual Coxeter number. 
The notion of boundary admissible level appeared in \cite{KW2} and at such a level the meromorphic continuation of characters has a particular nice form.
\end{rem}

\subsection{$(A_1, A_{2p-3})$ Argyres-Douglas theories}

The central  charge for this case is 
\[
c_p= 2-6\frac{\left(p-1\right)^2}{p}
\]
and equation (5.14) of \cite{BN} is easily modified to
\begin{equation}\label{eq:Aodd}
q^{-\frac{c_p}{24}}\mathcal I_{(A_1, A_{2p-3})}(z; q)  =  \frac{1}{\eta(q)^2} \sum_{n\in \mathbb Z}\left( \frac{q^{p\left(n+\frac{1}{2}-\frac{1}{2p}\right)^2}}{1-zq^{p\left(n+\frac{1}{2}-\frac{1}{2p}\right)}} - \frac{q^{p\left(n+\frac{1}{2}+\frac{1}{2p}\right)^2}}{1-zq^{p\left(n+\frac{1}{2}+\frac{1}{2p}\right)}}\right).
\end{equation}
The domain for $z$ and $q$ is $|z^{\pm 1}q^{\frac{p-1}{2}}| < 1$.
We will observe that this formula coincides with the vacuum character of the $\cB_p$-algebra of \cite{CRW}. This algebra is conjecturally a subregular $W$-algebra of $\mathfrak{sl}_{p-1}$ at level $k+p-1=\frac{p-1}{p}$ and this conjecture is true in the first cases $p=3, 4, 5$. The case $p=5$ follows due to the recent identification of Naoki Genra  \cite{G} of these subregular $W$-algebras with Feigin-Semikhatov algebras \cite{FS}.

\subsection{$(A_1, D_{2p})$ Argyres-Douglas theories}

The Schur-index of the $(A_1, D_{2})$  Argyres-Douglas theory is just the vacuum character of the rank two $\beta\gamma$-\voa{} $\cS^2$, which we will discuss as an instructive baby example in Section \ref{sec:D2}.

The central charges for this series are
\[
c_p = 4-6p.
\]
A closed formula for the Schur-Index is conjectured in equation (1.5) of \cite{BN}.
It reads
\begin{equation}\label{eq:Aodd}
\mathcal I_{(A_1, D_{2p})}(x, z; q)  =  \sum_{m=0}^\infty \tilde f^p_{\rho_m}(q; x) f_{\rho_m}(q; z)
\end{equation}
with $\rho_m$ the $m$-dimensional irreducible representation of $\sltwo$ and 
\[
\tilde f^p_{\rho_m}(q; x) = \frac{q^{\frac{p(m^2-1)}{4}}}{\prod_{n=1}^\infty (1-q^n)} \tr_{\rho_m}\left(x^{h} q^{-\frac{p h^2}{4}} \right) 
\]
where our $\sltwo$ convention is that $h, e, f$ form a basis with commutation relations
\[
[h, e]=2e, \qquad [h, f]=-2f, \qquad [e, f]=h.  
\]
$f_{\rho_m}(q;z)=P.E.\left[ \frac{q}{1-q}\text{ch}[\rho_3](z)\right] \ch[\rho_m](z)$ is defined using the plethystic exponential 
\[
P.E.[F(q; z)] = \text{exp}\left( \sum_{k=1}^\infty \frac{F(q^k; z^k)}{k} \right)
\]
and it can be simplified as follows
\begin{equation}\label{eq:PE}
\begin{split}
f_{\rho_m}(q;z) &= P.E.\left[ \frac{q}{1-q}\text{ch}[\rho_3](z)\right] \ch[\rho_m](z) \\
&= P.E.\left[ \frac{q}{1-q} \left(z^2 + 1 +z^{-2}\right)\right] \ch[\rho_m](z) \\
&= P.E.\left[ \sum_{n=1}^\infty q^n \left(z^2 + 1 +z^{-2}\right)\right] \ch[\rho_m](z) \\
&= \ch[\rho_m](z)\prod_{n=1}^\infty \text{exp}\left( q^n \left(z^2 + 1 +z^{-2}\right) \right) \\ 
&=  \ch[\rho_m](z)\prod_{n=1}^\infty  \left(1-z^2q^n\right)^{-1} \left(1-q^n\right)^{-1} \left(1-z^{-2}q^n\right)^{-1}.
\end{split}
\end{equation}

\section{The $L_k(\sltwo)\otimes \mathcal H$-module $\cX_p$}\label{sec:Xp}

The aim is to find a \voa{} for each Schur-index. For this, we will first find an $L_k(\sltwo)\otimes \mathcal H$-module $\cX_p$ whose character coincides with the $(A_1, D_{2p})$ Schur-index.

\subsection{The affine \voa{} $L_k(\sltwo)$}

I use \cite{CR1, CR2} for conventions on the affine \voa{} $L_k(\sltwo)$.

We denote the generators of $\mathfrak \sltwo$ by $h_n, e_n, f_n, K$ and $d$ for $n$ integer. As usual we identify the action of $d$ with $-L_0$. $K$ is central and the modes satisfy
\begin{equation}\nonumber
\begin{split}
[h_m, e_n] &= +2e_{n+m},\qquad [h_m, h_n] = 2m\delta_{n+m, 0}K,\qquad\qquad\quad [e_m, e_n] = 0,\\ 
[h_m, f_n] &= -2f_{n+m},\qquad [e_m, f_n] = h_{n+m} +m\delta_{n+m, 0}K,\qquad [f_m, f_n] = 0.
\end{split}
\end{equation}
The affine \voa{} $L_k(\sltwo)$ is strongly generated by fields $e(z), h(z), f(z)$ with operator products
\begin{equation}\nonumber
\begin{split}
h(z)e(w) &\sim \frac{2e(w)}{(z-w)}, \qquad h(z)h(w) \sim \frac{2k}{(z-w)^2}, \qquad h(z)f(w) \sim -\frac{2f(w)}{(z-w)}, \\
e(z)f(w)&\sim  \frac{k}{(z-w)^2},+\frac{h(w)}{(z-w)}.
\end{split}
\end{equation}
The Virasoro field for $k+2\neq 0$ is 
\[
L(z) = \frac{1}{2(k+2)} \left(\frac{1}{2} :h(z)h(z): + :e(z)f(z): + :f(z)e(z); \right)
\]
and it has central charge 
\[
c=\frac{3k}{k+2} = 3-\frac{6}{t}
\]
with $k+2=t$. 
Let 
\[
\mathfrak n = \text{span}_\CC\left( e_{n-1}, f_n, h_n \ | \ n\geq 1\right) 
\]
be a nilpotent subalgebra, 
\[
\mathfrak h =  \CC h_0 \oplus \CC K\oplus \CC d
\]
the Cartan subalgebra and $\mathfrak b=\mathfrak h\oplus \mathfrak n$ the corresponding Borel subalgebra. 
Let $\CC_\lambda$ the one-dimensional $\mathfrak b$-module on which $\mathfrak n$ acts as zero,  $h_0$ acts by multiplication with $\lambda$, $K$ by $k$ and $L_0$ by $h_\lambda$ 
where
\[
h_\lambda = \frac{\lambda(\lambda+2)}{4(k+2)}.
\]
Denote by
\[
V(k, \lambda) =  \text{Ind}^{\widehat\sltwo}_{\mathfrak b} \CC_\lambda
\]
and its simple quotient by $L(k, \lambda)$. We will surpress the level $k$ if its value is clear.
Singular vectors are determined by the determinant of the Shapovalov form on the weight spaces. 
This determinant for the weight space of weight $(\lambda-\mu, k, h_\lambda+m)$ is given by the formula of Kac and Kazhdan \cite{KK}, which I take from (2.8) of \cite{Ri}.
\begin{equation}\label{eq:KK}
\begin{split}
\text{det}_\lambda(\mu, m)&= \prod_{\ell, n=1}^\infty \left(\lambda+1-\ell\right)^{P\left(2\ell-\mu, m\right)}   \left(\lambda+1+n(k+2)-\ell\right)^{P\left(2\ell-\mu, m-n\ell\right)}   \\
&\qquad \cdot \left(-\lambda-1+n(k+2)-\ell\right)^{P\left(-2\ell-\mu, m-n\ell\right)}    \left(n(k+2)\right)^{P\left(-\mu, m-n\ell\right)}   
\end{split}
\end{equation}
$P(\mu, m)$ is the multiplicity of the weight $(\mu, m)$ in the level zero vacuum module $V_0(\sltwo)$.
A singular vector is a highest-weight vector of the Verma module and it appears if one of the factors vanishes as well as the arguments of $P$ appearing in the exponent of the corresponding factor.    

\subsubsection{The case $k+2=\frac{1}{p}$}\label{sec:Op}

For this work, the relevant levels are those for which $k+2=\frac{1}{p}$ and $p$ is a positive integer. 
In this case, the universal affine \voa{} is simple, i.e. $V_k(\sltwo)=L_k(\sltwo)$ \cite{GK}. For the purpose of this article it remains to understand irreducibility of all Weyl modules as suggested in \cite[Remark 6.2]{AKMP}.

We need to analyze the factors of \eqref{eq:KK} of $V(k, \lambda)$ for $\lambda>0$.
\begin{enumerate}
\item The first factor vanishes for $\ell =\lambda+1$ and the arguments of $P$ vanish if $\mu=2\ell=2(\lambda+1)$ and $m=0$. We thus have a singular vector of weight $(-\lambda-2, k, h_\lambda)$. 
 \item The second factor vanishes if $\lambda+1+\frac{n}{p}-\ell = 0$ and the argument of the corresponding $P$ vanish if $2\ell=\mu$ and $m=n\ell$. Set $n=pr$. The condition of being a singular vector constraints the conformal weight to satisfy $h_{\lambda-\mu}=h_\lambda+m$ implying that $4m=\mu p(\mu-2-\lambda)$. It follows that $2r\mu = 4r\ell = \mu (\mu-2-\lambda)$. Now $\mu=0$ implies $\ell=0$, which cannot be since $\ell>1$.  We thus have $2r=\mu-2-\lambda$ and thus 
 \[
 0=2\lambda + 2 +2r-2\ell = 2\lambda+2+\mu-2-\lambda-\mu=\lambda. 
 \]
 \item The third factor vanishes if $-\lambda-1+r-\ell=0$ as well as $2\ell-\mu$ and $m=n\ell= h_{\lambda-\mu}-h_\lambda$. An analogous computation as case two reveals that then $\lambda=0$. 
 \end{enumerate}
 We thus have the short-exact sequence for $\lambda \in \ZZ_{>0}$
 \[
 0 \longrightarrow V(k, -\lambda-2) \longrightarrow V(k, \lambda) \longrightarrow L(k, \lambda)  \longrightarrow 0. 
 \]
 Let $\cO$ be the category of level $k=-2+1/p$-modules with the following three properties: the conformal weight of any object is bounded below, $\mathfrak h$ acts semi-simply, $L_0$ eigenspaces are finite-dimensional.
 The simple objects of $\cO$ are then the Weyl modules $L(k, \lambda)$ for $\lambda$ non-negative integer. Note, that $L(k, 0)= V_k(\sltwo)$. This category is semi-simple \cite{K}.

Let $|z^{\pm 2}q|<1$, then the character of such a module is 
\begin{equation}
\begin{split}
\text{ch}[L(n+1)] (z;q)&=q^{\Delta_n+\frac{p}{4}-\frac{1}{p}}\prod_{n=0}^\infty \frac{\text{ch}[\rho_n](z)}{(1-z^2q^{n+1})(1-z^{-2}q^{n+1})(1-q^{n+1})}\end{split}
\end{equation}
where the conformal weight $\Delta_n$ is $p(n^2-1)/4$ and $\rho_n$ denotes the $n$-dimensional irreducible representation of $\sltwo$ whose character is 
\begin{equation}\nonumber
\begin{split}
\text{ch}[{\rho_{n}}](z) &=  z^{n-1}+z^{n-3}+...+z^{-n+3}+z^{-n+1} =
 \frac{z^{n-1}-z^{-n-1}}{1-z^{-2}}. 
\end{split}
\end{equation}

\subsubsection{Spectral flow}

A standard way of constructing new modules out of given ones is spectral flow. The notation here is taken from \cite{CR1, CR2}.
Define $\sigma^\ell$ for integer $\ell$ by
\[
\sigma^\ell(e_n)=e_{n-\ell}, \quad \sigma^\ell(h_n)=h_{n} -\delta_{n, 0}\ell k, \quad \sigma^\ell(f_n)=f_{n+\ell}, \quad \sigma^\ell(L_0)=L_0-\frac{1}{2} \ell h_0 +\frac{\ell^2}{4}k.
\]
The last equation is of course a consequence of the first three. Let $M$ be a module of $\widehat\sltwo$ on which $K$ acts by the scalar $k$. One then defines $\sigma^\ell(M)$ as follows. As a vector space it is isomorphic to $M$ with isomorphism denoted by $\sigma^{*, \ell}$ and the action of any $X$ in  $\widehat\sltwo$ is given by
\[
X \sigma^{*, \ell}(v) = \sigma^{*, \ell}(\sigma^{-\ell}(X) v). 
\]
Most important to this work is the relation of characters,
\[
\ch[\sigma^\ell(M)](z; q) =  q^{\frac{\ell^2}{4p}-\frac{\ell^2}{2}} z^{\frac{\ell}{p}-2\ell} \ch[M](zq^{\ell/2}; q).
 \]
 
\subsection{The $L_k(\sltwo)\otimes \mathcal H$-module $\cX_p$ and Appell-Lerch sums}

Fix a positive integer $p$.
Let $\mathcal H$ be the rank one Heisenberg \voa. We need the Fock modules $F_{\lambda_p n}$ of highest-weight $\lambda_p n$. Here $\lambda_p^2=-\frac{p}{2}$, so that the character is 
\[
\ch[F_{\lambda_p n}](x; q)) = \frac{x^n q^{-p \frac{n^2}{4}}}{\eta(q)}.
\]
Whenever it is clear that we have fixed a positive integer $p$ we will surpress the subscript $p$ of $\lambda_p$. 

Let $k+2=\frac{1}{p}$ so that the central charge of $L_k(\sltwo)\otimes \mathcal H$ is
\[
c_p = 4-6p.
\]
We are interested in the following $L_k(\sltwo)\otimes \mathcal H$-module
\begin{equation}\label{eq:Xp}
\cX_p :=  \bigoplus_{m=0}^\infty  L(m) \otimes \left( F_{-\lambda m} \oplus F_{-\lambda(m-2)} \oplus \dots \oplus F_{\lambda (m-2)} \oplus F_{\lambda m}\right).
\end{equation}
Plugging \eqref{eq:PE} into the formula for $\mathcal I_{(A_1, D_{2p})}(x, z; q)$ one immediately sees that
\begin{thm}\label{thm:XpSchur} The character of $\cX_p$ agrees with the Schur-index $\mathcal I_{(A_1, D_{2p})}(x, z; q)$
\begin{equation}
q^{-\frac{c_p}{24}}\mathcal I_{(A_1, D_{2p})}(x, z; q)  = \ch[\cX_p](x, z; q). 
\end{equation}
\end{thm}
The question is thus: Can $\cX_p$ be given the structure of a simple \voa?
We will provide a positive answer in Section \ref{sec:Walgebras}.

The level two Appell-Lerch sum of \cite{AC} is
\[
A_2(s, t, \tau) = w \sum_{n\in \ZZ} \frac{q^{n(n+1)}y^n}{1-wq^n}, \qquad w=e(s),\ y=e(t). 
\]
Its modular and elliptic properties are given in Corollary 7 of \cite{AC}:
\begin{lem}
The following properties hold
\begin{equation}\nonumber
\begin{split}
A_2(s+1, t; \tau) &= A_2(s, t; \tau) \\
A_2(s+\tau, t; \tau) &= w^2y^{-1}q A_2(s, t; \tau) + wy^{-\frac{1}{2}}q^{\frac{3}{4}}\sum_{m=0}^1w^mq^{\frac{m}{2}}\theta_1\left(t+m\tau +\frac{3}{2}; 2\tau\right)   \\
A_2(s, t+1; \tau) &= A_2(s, t; \tau) \\
A_2(s, t+\tau; \tau) &= w^{-1} A_2(s, t; \tau) - y^{-\frac{1}{2}}q^{-\frac{1}{4}}\theta_1\left(t +\frac{3}{2}; 2\tau\right)   \\
A_2(s, t; \tau+1) &= A_2(s, t; \tau) \\
A_2\left(\frac{s}{\tau}, \frac{t}{\tau}, -\frac{1}{\tau} \right) &= \tau e^{-2\pi i\left(s^2-st\right)/\tau} \left(A_2(s, t; \tau) + \phantom{\sum_1^1}\right. \\
&+\left. \frac{i}{2}\sum_{m=0}^{1} e^{2\pi i sm} \theta_1\left(t+m\tau-{\frac{1}{2}}; 2\tau\right) h\left(2s-t-m\tau+\frac{1}{2}; 2\tau\right) \right)
\end{split}
\end{equation}
with Mordell integral $h(s; \tau) = \int_\mathbb R\frac{e^{\pi i \tau x^2-2\pi sx}}{\text{cosh}(\pi x)} dx$ and standard Jacobi theta function 
$\theta_1(t; \tau) = -i\sum\limits_{n\in\ZZ} (-1)^n e^{\pi (n+1/2)^2+2\pi i t (n+1/2)}$.
\end{lem}
The combined transformation
\begin{equation}\label{eq:comb}
A_2(s\pm t +p\tau, \pm 2 (t+p\tau); \tau) =q^{\pm p}y^{\mp 2}A_2(s\pm t, \pm 2t; p\tau) 
\end{equation}
will be used in a moment.
Define 
\[
\AL_2(s, t; \tau) := A_2(s+t, 2t, \tau) - A_2(s-t, -2t, \tau). 
\]
Then using the following properties of $\theta_1(t, \tau)$ \cite{Z}:
\begin{equation}\nonumber
\theta_1(t+1; \tau)= -\theta_1(t; \tau), \qquad \theta_1(t+\tau; \tau) = -e^{-\pi i \tau-2\pi i t} \theta_1(t; \tau), 
\qquad \theta_1(-t; \tau) = -\theta_1(t; \tau),
\end{equation}
it is straightforward to verify that $AL_2$ transforms as a Jacobi form
\begin{cor}
$AL_2$ satisfies the elliptic and modular transformation properties
\begin{equation}\nonumber
\begin{split}
\AL_2(s+1, t; \tau) &= \AL_2(s, t; \tau) \\
\AL_2(s+\tau, t; \tau) &= w^2q \AL_2(s, t; \tau)    \\
\AL_2(s, t+1; \tau) &= \AL_2(s, t; \tau) \\
\AL_2(s, t+\tau; \tau) &= y^{-2}q \AL_2(s, t; \tau)   \\
\AL_2(s, t; \tau+1) &= \AL_2(s, t; \tau) \\
\AL_2\left(\frac{s}{\tau}, \frac{t}{\tau}, -\frac{1}{\tau} \right) &= \tau e^{-2\pi i\left(s^2-st\right)/\tau} \AL_2(s, t; \tau). 
\end{split}
\end{equation}
\end{cor}
We also need the computation
\begin{equation}\nonumber
\begin{split}
A_2\left(\frac{s+t}{p},  \frac{2t}{p}; \frac{\tau}{p}\right) &= \sum_{n\in\ZZ} \frac{q^{\frac{n(n+1)}{p}} y^{\frac{2n+1}{p}}w^{\frac{1}{p}}}{1-w^{\frac{1}{p}}y^{\frac{1}{p}}q^{\frac{n}{p}}}
= \sum_{a=0}^{p-1} \sum_{n\in\ZZ} \frac{q^{\frac{n(n+1)}{p}} y^{\frac{2n+a+1}{p}}q^{\frac{an}{p}}w^{\frac{a+1}{p}}}{1-wyq^{n}}\\
&= \sum_{a=0}^{p-1}\sum_{b=0}^{p-1} \sum_{n\in\ZZ} \frac{q^{pn(n+1)}  q^{n( 2b+1-p)} q^\frac{b(b+ 1)}{p} y^{ \left(2n + \frac{2b+a+1}{p}\right)}q^{an+\frac{ab}{p}}w^{\frac{a+1}{p}}}{1-wyq^{pn + b}}\\
&\hspace{-2cm}= \sum_{a=0}^{p-1}\sum_{b=0}^{p-1} y^{\frac{2b+a+1-p}{p}}w^{\frac{a+1-p}{p}} q^{\frac{b}{p}(b+a+1-p)} A_2(s+t+ b\tau,  2t+(2b+a+1-p)\tau; p\tau).
\end{split}
\end{equation}
and very similarly
\begin{equation}\nonumber
\begin{split}
A_2\left(\frac{s-t}{p},  -\frac{2t}{p}; \frac{\tau}{p}\right) &=  \\
&\hspace{-2.6cm}= \sum_{a=0}^{p-1}\sum_{b=0}^{p-1} y^{\frac{2b-a-1+p}{p}}w^{\frac{a+1-p}{p}} q^{\frac{b}{p}(b-a-1+p)} A_2(s-t-b\tau,  -2t-(2b-a-1+p)\tau; p\tau).
\end{split}
\end{equation}
Now, we observe using \eqref{eq:comb} that the interval of summation in the variable $b$ can be shifted by any integer multiple of $p$.
With this preparation we get the modular transformation properties of characters. 

Define 
\[
\Pi(z) = q^{\frac{1}{6}}(z-z^{-1}) \prod_{n=1}^\infty (1-z^2q^n)(1-q^n)^2(1-z^{-2}q^n)
\]
\begin{lem}
Let  $|x|^2 < |z|^{\pm 2}< |q|^{-1}$ then the character satisfies
\begin{equation}\nonumber 
\begin{split}
\text{ch}[\cX_p] &=\frac{q^{\frac{p}{4}}}{\Pi(z)}   \sum_{m\in\mathbb Z} \left( \frac{q^{pm(m+1)}z^{2m+1}}{1-xzq^{pm+\frac{p}{2}}} -\frac{q^{pm(m+1)}z^{-2m-1}}{1-xz^{-1}q^{pm+\frac{p}{2}}}\right) 
= \frac{x^{-1}q^{-\frac{p}{4}}}{\Pi(z)} \AL_2\left(u+\frac{p\tau}{2}, v; p\tau\right).
\end{split}
\end{equation}
\end{lem}
\begin{proof}
We compute
\begin{equation}\nonumber 
\begin{split}
q^{-\frac{p}{4}} \Pi(z) \ch[\cX_p]  &= \sum_{a = 0}^\infty \sum_{m=0}^\infty x^a q^{-\frac{p}{4}a^2} q^{\frac{p}{4} \left( 2m+a\right) \left( 2(m+1)+a\right)} \left( z^{2m+a+1} - z^{-2m-a-1} \right) + \\
&\quad  \sum_{a = -1}^{-\infty} \sum_{m=0}^\infty x^a q^{-\frac{p}{4}a^2} q^{\frac{p}{4} \left( 2m-a \right) \left( 2(m+1)-a\right)} \left( z^{2m+a+1} - z^{-2m-a-1} \right)  \\
&= \sum_{a = 0}^\infty \sum_{m=0}^\infty x^a  q^{p m\left( m+1\right) +\frac{p}{2} a\left( 2m+1\right)} \left( z^{2m+a+1} - z^{-2m-a-1} \right) + \\
&\quad  \sum_{a = -1}^{-\infty} \sum_{m=0}^\infty x^a  q^{p m\left( m+1\right) -\frac{p}{2} a\left( 2m+1\right)} \left( z^{2m+a+1} - z^{-2m-a-1} 
\right)\\
&= \sum_{a = 0}^\infty \sum_{m=0}^\infty x^a  q^{p m\left( m+1\right) +\frac{p}{2} a\left( 2m+1\right)} \left( z^{2m+a+1} - z^{-2m-a-1} \right) + \\
&\quad  \sum_{a = -1}^{-\infty} \sum_{m=-1}^{-\infty} x^a  q^{p m\left( m+1\right) +\frac{p}{2} a\left( 2m+1\right)} \left( z^{-2m+a-1} - z^{2m-a+1} \right).
\end{split}
\end{equation}
Performing the summation over $a$ then gives the claim. 
\end{proof}
Recall that for  $\ell$ in $\ZZ$ one has spectral flow actions on modules, it acts on characters as 
\[
\ch[\sigma^\ell(M)](z; q) =  q^{\frac{\ell^2}{4p}-\frac{\ell^2}{2}} z^{\frac{\ell}{p}-2\ell} \ch[M](zq^{\ell/2}; q)
 \]
and for Heisenberg modules it is a simple shift in the weight label, i.e. 
\[
\ch[F_{\lambda\left(n-\frac{\ell'}{p}\right)}](x; q) = q^{-\frac{\ell'^2}{4p}}x^{-\frac{\ell'}{p}}\ch[F_{\lambda n}]\left(xq^{\ell'/2}; q\right).
\]
If the character of a module converges for $(x, y, q)$ in $D \subset \CC^2\times \mathbb H$, then the spectrally flown one converges for 
$(xq^{\ell'/2}, zq^{\ell/2}, q)$ in $D$. 
We define 
\[
\sigma^{(\ell, \ell')}\left(\cX_p\right) := \bigoplus_{m=0}^\infty \bigoplus_{s=0}^{m} \sigma^{\ell}\left(L(m)\right) \otimes F_{-\lambda \left(m-2s-\frac{\ell'}{p}\right)}.
\]
If $p$ is odd, then we also define the super character as
\[
\ch^-[\sigma^{(\ell, \ell')}\left(\cX_p\right)](u, v; \tau) = \ch[\sigma^{(\ell, \ell')}\left(\cX_p\right)](u+\frac{1}{2}, v; \tau).
\]
From now on, we view characters as meromorphic functions on $\CC^2\times \mathbb H$.
Define the set $S_p := \{ (\ell, \ell')\ |\ -p \leq \ell \leq p-1,\ 0\leq \ell' \leq p-1, \ \ell+\ell'+p \in 2\ZZ\}$. Note, that in the case $p$ odd this set parameterizes twisted modules. The set of local modules is
\[
S^{\text{loc}}_p := \{ (\ell, \ell')\ |\ -p \leq \ell \leq p-1,\ 0\leq \ell' \leq p-1, \ \ell+\ell' \in 2\ZZ\}.
\]
\begin{rem}
We will later prove that $\cX_p$ can be given the structure of a \voa. In the case of $p$ even this is an ordinary integer graded \voa. If $p$ is odd it is half-integer graded and thus of wrong statistics. In this case the integer graded modules are twisted and the ones that are half-integer but not integer graded are local
modules of the \voa{} \cite{CKM}. This is the justification for above naming. 
\end{rem}
\begin{thm}\label{thm:modular}
Denote by $\ch_p$ the character for $p$ even and the supercharacter for $p$ odd, then as meromorphic functions on $\CC^2\times \mathbb H$
\begin{equation}\nonumber
\begin{split}
\text{ch}[\cX_p]\left(\frac{u}{\tau}, \frac{v}{\tau}, -\frac{1}{\tau} \right) &=   e^{\frac{2\pi i }{\tau} \left( kv^2-\frac{u^2}{p} \right)}   \sum_{(\ell, \ell')\in S_p} \ch_p[\sigma^{(\ell, \ell')}(\cX_p)](u, v; \tau)
\end{split}
\end{equation}
and for $(n, n')$ in $S^{\text{loc}}_p$.
\begin{equation}\nonumber
\begin{split}
\text{ch}[\sigma^{(n, n')}(\cX_p)]\left(\frac{u}{\tau}, \frac{v}{\tau}, -\frac{1}{\tau} \right) &=   e^{\frac{2\pi i }{\tau} \left( kv^2-\frac{u^2}{p} \right)}    \sum_{(\ell, \ell')\in S_p} S_{(n, n')(\ell, \ell')} \ch_p[\sigma^{(\ell, \ell')}(\cX_p)](u, v; \tau)
\end{split}
\end{equation}
with $S$-matrix entries $S_{(n, n')(\ell, \ell')} =-\frac{e^{-\frac{2\pi i}{p} \left(n\ell-n'\ell'\right)}}{p}$.
\end{thm}
\begin{proof}
Let $z=e(v), x=e(u)$.
For computational convenience define
\[
\mu_{a, b}\left(\cX_p\right) := \bigoplus_{m=0}^\infty \bigoplus_{s=0}^{m} \sigma^{2b+a}\left(L(m)\right) \otimes F_{-\lambda \left(m-2s-\frac{a+p}{p}\right)}.
\]
A short computation reveals that in their domain of convergence
\begin{equation}\nonumber
\begin{split}
\ch[\mu_{a, b}\left(\cX_p\right)](u, v; \tau) &= \frac{1}{\Pi(v; \tau)} \left( x^{\frac{a}{p}}z^{\frac{2b+a}{p}} q^{\frac{b}{p}(b+a)} A_2(u+v+b\tau, 2v+2b+a, p\tau) \right.- \\
&\qquad\qquad \left. x^{\frac{a}{p}}z^{-\frac{a+2b}{p}} q^{\frac{b}{p}(b+a)} A_2(u-v-b\tau, -2v-2b-a, p\tau) \right). 
\end{split}
\end{equation}
The first equality of the Theorem follows immediately from our discussion of $\AL_2$. The second equaltiy of the theorem follows directly from the first one replacing $(u, v)$ by $(u+n'\tau, v+n\tau)$. 
\end{proof}

\section{The $\cB_p$-algebra and $A_{2p-3}$ Argyres-Douglas theories}\label{sec:Bp}

There are three closely related families of \voas, the singlet algebra $\Sing$, the triplet algebra $\Trip$ and the logarithmic $\cB_p$-algebra. 
The parameter $p$ is a positie integer and at least two. 
The first two \voas{} are the best understood logarithmic \voas. The word logarithmic is indicating that correlation functions might have logarithmic singularities reflecting that modules are not necessarily completely reducible. The name singlet refers to that $\Sing$ is strongly generated by the Virasoro field and one more field of conformal dimension $2p-1$. In the triplet case, there are three such additional weight $2p-1$ generators. The triplet is an infinite order simple current extension of the singlet and the singlet is a coset of $\cB_p$. 
References on these \voas{} are \cite{TW, AM1, AM2} for the triplet, \cite{CM, CMR, AM3} for the single and $\cB_p$ has been introduced in \cite{CRW}, though the cases $p=2, 3$ appeared also in \cite{Ri2, CR3, A3}. I follow \cite{CRW} which used mainly these mentioned references.

\subsection{The singlet $\Sing$ and triplet $\Trip$ algebra}

The singlet algebra has central charge 
\[
c_p = 1-6\frac{(p-1)^2}{p}.
\]
It is an extension of the simple Virasoro algebra at that central charge and as a Virasoro module
\begin{equation} \label{eq:Mpdecomp}
\Sing \cong \bigoplus_{n=0}^\infty L(h_{2n+1, 1}, c_p)
\end{equation}
where $L(h_{r, s}, c_p)$ denotes the simple Virasoro lowest-weight module of weight 
\[
h_{r, s} = \frac{\alpha_{r, s}(\alpha_{r, s}-\alpha_0)}{2}, \qquad \alpha_{r, s} = \frac{1-r}{2}\alpha_+ +  \frac{1s}{2}\alpha_-
\]
and parameters $\alpha_+=\sqrt{2p}, \alpha_-=-\sqrt{2/p}$ and $\alpha_0=\alpha_++\alpha_-$. 
The simple modules fall into two types, atypical and typical ones. We only need the atypical modules $M_{r, s}$. As Virasoro modules they decompose as
\begin{equation}
\begin{split}
M_{r,s} &\cong   \bigoplus_{k=0}^\infty L(h_{r+2k,s}, c_p ), \qquad  r\geq 1, 1\leq s\leq p \\
 M_{r+1,p-s} &\cong   \bigoplus_{k=0}^\infty L(h_{r-2k,s}, c_p ), \qquad  r\leq 0, 1\leq s\leq p.
\end{split}
\end{equation}
Note, that $M_{r, s}$ and $M_{2-r, s}$ are isomorphic as Virasoro modules but not as singlet modules. Their characters are given in terms of false theta functions \cite{CM}
\[
\text{ch}[M_{r, s}] =\frac{1}{\eta(q)} \sum_{n\geq 0} q^{p\left(\frac{r}{2}+n-\frac{s}{2p}\right)^2} - q^{p\left(\frac{r}{2}+n+\frac{s}{2p}\right)^2}
\]
for $r\geq 1$ and for $r\leq 0$ one has
\[
\text{ch}[M_{r, s}] =\frac{1}{\eta(q)} \sum_{n\leq 0} q^{p\left(\frac{r}{2}+n-1+\frac{s}{2p}\right)^2} - q^{p\left(\frac{r}{2}+n-1-\frac{s}{2p}\right)^2}
\]

The triplet algebra is an infinite order simple current extension of the singlet  algebra. We only need its decomposition into Virasoro modules
\begin{equation} \label{eq:Wpdecomp}
\Trip \cong \bigoplus_{n=0}^\infty (2n+1) L(h_{2n+1, 1}, c_p).
\end{equation}
This decomposition suggests an action of SU$(2)$ on $\Trip$. It actually happens that PSL$(2, \CC)$ acts on the triplet algebra \cite{ALM}, so that as a PSL$(2, \CC)\, \otimes$ Vir-module
\begin{equation}
\Trip \cong  \bigoplus_{n=0}^\infty \rho_{2n+1} \otimes L(h_{2n+1, 1}, c_p)
\end{equation}
with $\rho_{n}$ the $n$-dimensional irreducible representation of PSL$(2, \CC)$.

\subsection{The logarithmic $\cB_p$-algebra}

The aim here is to compute the vacuum characters of the $\cB_p$ algebras constructed in \cite{CRW}. I keep this section short as a general study of $\cB_p$ and its representations is work in progress \cite{ACKR}. These $\cB_p$-algebras are conjecturally the subregular quantum Hamiltonian reductions of $\mathfrak{sl}_{p-1}$ at level $k=-(p-1)^2/p$. Note, that this is a boundary admissible level:
\[
k+p-1= \frac{p-1}{p}.
\] 
This conjecture is true for $p=  3, 4, 5$. The last case due to the identification of the Feigin-Semikhatov algebras with these quantum Hamiltonian reductions \cite{G}. The case $p=2$ is the $\beta\gamma$ \voa{} of rank one. 
  
We again use the set-up and notation of \cite{CRW}. Then one result of that article is the construction and identification with above W-algebras as infinite order simple current extensions of $\Sing\otimes \mathcal H$ of the form
\begin{equation} \label{eq:Bpdecomp}
\begin{split}
\mathcal B_p &= \bigoplus_{r\in\mathbb Z} M_{r+1, 1} \otimes F_{\lambda_p  r} \\
&\cong \bigoplus\limits_{m=0}^\infty  L(h_{m, 1}, c_p) \otimes \left( F_{-\lambda_p m} \oplus F_{-\lambda_p(m-2)} \oplus \dots \oplus F_{\lambda_p (m-2)} \oplus F_{\lambda_p m}\right).
\end{split}
\end{equation}
Recall the character of the Fock module is simply
\[
\text{ch}\left[ F_{\lambda_p  r}\right] =\frac{x^rq^{-p\frac{r^2}{4}}}{\eta(q)}. 
\]
One can now construct modules of $\mathcal B_p$ either directly from the realization of $\cB_p$ as kernel of screenings of a certain lattice \voa-module or using the theory of \voa{} extensions for simple currents \cite{CKM, CKL}.   
The following are induced objects (local if $s$ is odd and twisted if $s$ is even)
\[
W_{s} := \bigoplus_{r\in\mathbb Z} M_{r+1, s} \otimes F_{\frac{r\alpha_+}{2}}.
\]
The case $s=1$ is $\mathcal B_p$.
One gets
\begin{equation}\nonumber
\begin{split}
\text{ch}[W_{s}](u; \tau)  &= \frac{1}{\eta(q)^2} \sum_{n, r\geq 0}\left( q^{p\left(n+\frac{1}{2}-\frac{s}{2p}\right)^2}x^rq^{pr\left(n+\frac{1}{2}-\frac{s}{2p}\right)} - q^{p\left(n+\frac{1}{2}+\frac{s}{2p}\right)^2}x^rq^{pr\left(n+\frac{1}{2}+\frac{s}{2p}\right)}\right)+\\
 &\quad \frac{1}{\eta(q)^2} \sum_{n, r\leq -1}\left( q^{p\left(n+\frac{1}{2}+\frac{s}{2p}\right)^2}x^rq^{pr\left(n+\frac{1}{2}+\frac{s}{2p}\right)} - q^{p\left(n+\frac{1}{2}-\frac{s}{2p}\right)^2}x^rq^{pr\left(n+\frac{1}{2}-\frac{s}{2p}\right)}\right) 
\end{split}
\end{equation}
for $x=e(u)$.
This identity is as formal power series. We see that for $|x| <1$ we can expand as geometric series to get
\begin{equation}\nonumber
\begin{split}
\text{ch}[W_{s}](u; \tau)  &= \frac{1}{\eta(q)^2} \sum_{n\in \mathbb Z}\left( \frac{q^{p\left(n+\frac{1}{2}-\frac{s}{2p}\right)^2}}{1-xq^{p\left(n+\frac{1}{2}-\frac{s}{2p}\right)}} - \frac{q^{p\left(n+\frac{1}{2}+\frac{s}{2p}\right)^2}}{1-xq^{p\left(n+\frac{1}{2}+\frac{s}{2p}\right)}}\right).
\end{split}
\end{equation}
Note, that 
\[
\ch[W_s](u; \tau) = \frac{\Pi(s\tau/2; \tau)}{\eta(\tau)^2}q^{\frac{s^2}{4p}} \ch[\cX_p](u; s\tau/2; \tau) = \lim_{z\rightarrow 1} \frac{\Pi(z; \tau)}{\eta(\tau)^2} \ch[\sigma^{s, 0}\cX_p](u; z; \tau).
\]
One can now get many more characters and their modules by using spectral flow. I.e.
\[
\sigma^{s'}(W_{s}) := \bigoplus_{r\in\mathbb Z} M_{r+1, s} \otimes F_{\lambda\left(r-\frac{s'}{p}\right)}.
\]
so that ($z=e(v)$)
\[
\ch[\sigma^{s'}(W_s)](u; \tau) = q^{-\frac{s'^2}{4p}} x^{-\frac{s'}{p}} \ch[W_s](u+{\frac{s'}{2}}\tau; \tau)= \lim_{z\rightarrow 1} \frac{\Pi(z; \tau)}{\eta(\tau)^2} \ch[\sigma^{s, s'}\cX_p](u; v; \tau).
\]
Especially the modular properties of $\AL_2$ derived in last section can be applied to list the modular properties of the meromorphic continuations of characters. 

Most importantly, we compare with \eqref{eq:Aodd}  ((5.14) of \cite{BN}) and see that 
\begin{thm}\label{thm:BPSchur}
\begin{equation}\nonumber
\text{ch}[\mathcal B_p] = q^{-\frac{c_p}{24}} \mathcal I_{A_1, A_{2p-3}}\left(q; z^{-1}\right)
\end{equation}
with 
\[
c_p= 2-6\frac{\left(p-1\right)^2}{p}
\]
the central charge of the $\mathcal B_p$-algebra.  
\end{thm}

Understanding the full representation categories of \voas{} is usually a very difficult question. Using the relation to singlet algebra allows to obtain much of the structure of the representation theory of the $\cB_p$-algebra. 

The $\cB_p$-algebra is an infinite order simple current extension of the singlet and Heisenberg algebra, $\Sing\otimes \mathcal H$.
One can thus use the conjectural equivalence of mod-$\Sing$ to a module category of the so-called restricted unrolled quantum group of $\sltwo$ \cite{CMR} together with the theory of \voa{} -extensions \cite{CKL, CKM} to understand the representation category of $\cB_p$. This is work in progress \cite{ACKR}. 

\section{Vertex algebras for $(A_1, D_{2p})$ Argyres-Douglas theories}\label{sec:Walgebras}

\subsection{ The \voa{} for $D_{2}$ Argyres-Douglas theory}\label{sec:D2}

It turs out to be useful to study this example first. It corresponds to fixing $p=1$ and $\lambda=\lambda_1$ satisfies $\lambda^2=-1/2$ in the conventions for $L_{-1}(\sltwo)$ and the Heisenberg \voa{} $\mathcal H$.
The idea for the following is based on \cite{BCR}. As a consequence of the results of this section we will be able to deduce our main result, that is Theorem \ref{thm:main}.

Let $\cS^2=\mathcal S \otimes \mathcal S$ be the \voa{} of two copies of the $\beta\gamma$-ghost voa $\mathcal S$. 
This \voa{} is strongly generated by even fields $\beta_1,\beta_2,\gamma_1,\gamma_2$ of conformal weight $1/2$ 
with operator products
\begin{equation}\nonumber
 \beta_i(z)\beta_j(w) \sim \gamma_i(z)\gamma_j(w) \sim 0 \quad,\quad \beta_i(z)\gamma_j(w)\sim \delta_{i,j}(z-w)^{-1}.
 \end{equation}
The bilinears $:\beta_i\gamma_j:$ are a homomorphic image of the universal affine vertex algebra of $\mathfrak{gl}_2$ at level $-1$.
Assign charges corresponding to the zero mode of the Heisenberg subalgebra $a(z)= :\beta_1(z)\gamma_1(z):+:\beta_2(z)\gamma_2(z):$
and to the current associated to the Cartan subalgebra of $\sltwo$, $h(z)= :\beta_1(z)\gamma_1(z):-:\beta_2(z)\gamma_2(z):,$ 
so that
\begin{equation}\nonumber
\begin{split}
 [a_0,\beta_i(z)]&=-\beta_i(z),\qquad [a_0,\gamma_i(z)]=\gamma_i(z),\qquad \ \ \,
 [h_0,\beta_1(z)]=-\beta_1(z),\\
 [h_0,\gamma_1(z)]&=\gamma_1(z),\qquad \ \, 
[h_0,\beta_2(z)]=\beta_2(z),\qquad\ \  [h_0,\gamma_2(z)]=-\gamma_2(z).
\end{split}
\end{equation}
Then, we have since $\cS^2$ is freely generated by those fields
\begin{equation}\nonumber
\begin{split}
\text{ch}[\cS^2](x,z;q)&= \tr_{\cS^2}\bigl(q^{L_0-\frac{c}{24}}z^{h_0}x^{a_0}\bigr)\\
&= q^{\frac{2}{24}}\prod_{n=0}^\infty \frac{1}{(1-xzq^{n+\frac{1}{2}})(1-xz^{-1}q^{n+\frac{1}{2}})(1-x^{-1}zq^{n+\frac{1}{2}})(1-x^{-1}z^{-1}q^{n+\frac{1}{2}})}
\end{split}
\end{equation}
in the domain
\[ |q^{\frac{1}{2}}|< |x|,|z| < |q^{-\frac{1}{2}}|.\] 
We want to compute the character of the commutant or coset subalgebra Com$(\mathcal H, \cS^2)$ of the Heisenberg algebra $\mathcal H$ generated by $a(z)$ in $\cS^2$. 
This essentially amounts to computing the Fourier coefficients in $x$. 
The commutant contains affine $\sltwo$ at level minus one, $L_{-1}(\sltwo)$, as subalgebra. Its vacuum character is
\begin{equation}\nonumber
\begin{split}
\text{ch}[L_{-1}(\sltwo)](z;q)&= \tr_V\bigl(q^{L_0-\frac{c}{24}}z^{h_0}\bigr)
= q^{\frac{3}{24}}\prod_{n=0}^\infty \frac{1}{(1-z^2q^{n+1})(1-z^{-2}q^{n+1})(1-q^{n+1})}.
\end{split}
\end{equation}
Note that the $n$-th Fourier coefficients in $x$ must be a character of $L_{-1}(\sltwo) \otimes \mathcal H$ with Heisenberg character of the form $q^{-n^2/4}/\eta(q)$. 
Since $L_0$-eigenspaces of $\cS^2$ are finite dimensional,
the same must be true for the $L_{-1}(\sltwo)$-modules appearing in the decomposition, i.e. they are in the category $\mathcal O_1$ introduced in Section \ref{sec:Op}. This means that Fourier coefficients have to be sums of the $\text{ch}[L(n)](z;q)$ and recall that
\begin{equation}\nonumber
\begin{split}
\text{ch}[L(n+1)] (z;q)&=q^{\Delta_n}\prod_{n=0}^\infty \frac{\text{ch}[\rho_n](z)}{(1-z^2q^{n+1})(1-z^{-2}q^{n+1})(1-q^{n+1})}\\
&= q^{\Delta_n-\frac{3}{24}}\text{ch}[\rho_n](z)\text{ch}[L_{-1}(\sltwo)](z;q)
\end{split}
\end{equation}
where the conformal weight $\Delta_n$ is $(n^2-1)/4$.
This strongly suggests to first decompose
\[ \frac{\text{ch}[\cS^2](x,z;q)}{\text{ch}[L_{-1}(\sltwo)](z;q)}. \]
This expression is up to a factor of $\eta(q)$ and a trivial renaming the product side of the denominator
identity of affine $\mathfrak{sl}(2|1)$ (see \cite{KW1}). It thus gives the identity
\begin{equation}\nonumber
 \frac{\text{ch}[\cS^2](x,z;q)}{\text{ch}[L_{-1}(\sltwo)](z;q)}= 
\frac{1}{\eta(q)}\frac{1}{1-z^{-2}} \sum_{j\in\Z} \frac{(zxq^{\frac{1}{2}})^j}{1-zx^{-1}q^{j+\frac{1}{2}}} 
\end{equation}
and we have to remember that $|q^{\frac{1}{2}}|< |x|,|z| < |q^{-\frac{1}{2}}|$.
For the moment, we also require that $z\neq 1$, though
we will later see that the limit exists.
We can expand 
\begin{equation}\nonumber
 \frac{1}{1-zx^{-1}q^{j+\frac{1}{2}}} = \begin{cases}
                                               \ \ \ \sum_{m=0}^\infty   \bigl(zx^{-1}q^{j+\frac{1}{2}}\bigr)^m & \text{if}\ j>0 \\
                                                -\sum_{m=1}^\infty   \bigl(zx^{-1}q^{j+\frac{1}{2}}\bigr)^{-m} & \text{if}\ j<0
                                               \end{cases}
 \end{equation}
and thus get 
\begin{equation}\nonumber
\begin{split}
 \frac{\text{ch}[\cS^2](x,z;q)}{\text{ch}[L_{-1}(\sltwo)](z;q)}&= 
\frac{1}{\eta(q)}\sum_{s\in\Z}x^s \sum_{m=0}^\infty q^{m(m+|s|+1)+\frac{|s|}{2}}\bigl(\frac{z^{2m+|s|}-z^{-2m-2-|s|}}{1-z^{-2}}  \bigr) \\
&= \frac{1}{\eta(q)}\sum_{s\in\Z}x^s \sum_{m=0}^\infty q^{m(m+|s|+1)+\frac{|s|}{2}} \text{ch}[{\rho_{2m+|s|+1}}](z)
\end{split}
\end{equation}
so that we arrive at the decomposition
\begin{equation}\nonumber
\begin{split}
 \text{ch}[\cS^2](x,z;q) &= \sum_{s\in\Z} \frac{x^s q^{-\frac{s^2}{4}}}{\eta(q)}\sum_{m=0}^\infty q^{m(m+|s|+1)+\frac{|s|}{2}+\frac{s^2}{4}}
 \text{ch}[{\rho_{2m+|s|+1}}](z)\text{ch}[L_{-1}(\sltwo)](z;q)\\
 &= \sum_{s\in\Z} \frac{x^s q^{-\frac{s^2}{4}}}{\eta(q)}\sum_{m=0}^\infty q^{\frac{(2m+|s|+1)^2}{4}-\frac{1}{4}}
 \text{ch}[{\rho_{2m+|s|+1}}](z)\text{ch}[L_{-1}(\sltwo)](z;q).
\end{split}
 \end{equation}
 For completenes let us check the limit $z\rightarrow 1$.
In this case, 
we have 
$ \text{ch}[{\rho_{2m+|s|+1}}](1)=2m+|s|+1$ and $q^{-1/4}\text{ch}[L_{-1}[\sltwo](1;q)=\eta(q)^{-3}$
so that our decomposition simplifies to
\begin{equation}\nonumber
\begin{split}
 \text{ch}[\cS^2](x,1;q) &= \sum_{s\in\Z} \frac{x^s q^{-\frac{s^2}{4}}}{\eta(q)}
 \sum_{m=0}^\infty (2m+|s|+1)q^{\frac{(2m+|s|+1)^2}{4}}\frac{1}{\eta(q)^3}
 \end{split}
 \end{equation}
we recognize the Fourier coefficients as a derivatives of partial theta functions
\begin{equation}\nonumber
 \frac{1}{2\pi i } \frac{d}{vu}P_a(v;\tau)\Big\vert_{v=0}, 
 \qquad P_a(v;\tau):=\sum_{m=0}^\infty e^{2\pi i v(2m+a+1)}q^{\frac{(2m+a+1)^2}{4}}.
\end{equation}
In summary, we have the following representation theoretic interpretation of the meromorphic Jacobi form decomposition problem.
\begin{prop}\label{thm:n=2}
 The character of $\cS^2$ decomposes into characters of $L_{-1}(\sltwo)\otimes \mathcal H$ as
 \begin{equation}\nonumber
\begin{split}
 \emph{ch}[\cS^2](x,z;q) &= \sum_{s\in\Z}  \emph{ch}[F_{\lambda s}](x;q)\sum_{m=0}^\infty \emph{ch}[L(2m+|s|)](z;q)
\end{split}
 \end{equation}
 and since $\mathcal O_1$ is semi-simple and characters are linearly independent we have
$\cS^2 \cong \mathcal X_1$ as $L_{-1}(\sltwo)\otimes \mathcal H$-modules.
\end{prop}
This decomposition had already been mentioned in \cite[Remark 3.3]{KW4}, but without proof.
The character identity rephrazes the denominator idenity of $\mathfrak{sl}(2|1)$ as follows
\begin{equation}\nonumber
\begin{split}
D(x, z; q) :&=  \prod_{n=0}^\infty  \frac{(1-z^2q^{n+1})(1-q^{n+1})^2(1-z^{-2}q^{n+1})}{(1-xzq^{n+\frac{1}{2}})(1-xz^{-1}q^{n+\frac{1}{2}})(1-x^{-1}zq^{n+\frac{1}{2}})(1-x^{-1}z^{-1}q^{n+\frac{1}{2}})} \\
&=\sum_{s\in\Z}  x^sq^{-\frac{s^2}{4}} \sum_{m=0}^\infty {\ch}[\rho_{2m+|s|+1}](z) q^{\frac{(2m+|s|+1)^2-1}{4}}. 
\end{split}
\end{equation}
This relates now nicely to the character of $\cX_p$ by rescaling $\tau\mapsto p\tau$, i.e.
define
 \begin{equation}\nonumber
\begin{split}
\text{Prod}_p(x, z; q)  &:= q^{\frac{p}{4}-\frac{1}{6}} \prod_{n=0}^\infty \frac{1}{(1-z^2q^{n+1})(1-q^{n+1})^2(1-z^{-2}q^{n+1})}  \times \\
& \hspace{-2.5cm}\times \prod_{n=0}^\infty  \frac{\left(1-z^2q^{p(n+1)}\right)\left(1-q^{p(n+1)}\right)^2\left(1-z^{-2}q^{p(n+1)}\right)}{\left(1-xzq^{p\left(n+\frac{1}{2}\right)}\right)\left(1-xz^{-1}q^{p\left(n+\frac{1}{2}\right)}\right)\left(1-x^{-1}zq^{p\left(n+\frac{1}{2}\right)}\right)\left(1-x^{-1}z^{-1}q^{p\left(n+\frac{1}{2}\right)}\right)}. 
\end{split}
\end{equation}
\begin{prop}\label{prop:main}
\[
\emph{Prod}_p(x, z; q)  = \ch\left[ \cX_p \right](x, z; q) 
\]
\end{prop}
\begin{proof}
The claim follows from the short computation:
\begin{equation}\nonumber
\begin{split}
\text{Prod}_p(x, z; q)   &= q^{\frac{p}{4}-\frac{1}{6}} \prod_{n=0}^\infty \frac{1}{(1-z^2q^{n+1})(1-q^{n+1})^2(1-z^{-2}q^{n+1})}  \times \\
& \hspace{-2.5cm}\times \prod_{n=0}^\infty  \frac{\left(1-z^2q^{p(n+1)}\right)\left(1-q^{p(n+1)}\right)^2\left(1-z^{-2}q^{p(n+1)}\right)}{\left(1-xzq^{p\left(n+\frac{1}{2}\right)}\right)\left(1-xz^{-1}q^{p\left(n+\frac{1}{2}\right)}\right)\left(1-x^{-1}zq^{p\left(n+\frac{1}{2}\right)}\right)\left(1-x^{-1}z^{-1}q^{p\left(n+\frac{1}{2}\right)}\right)} \\
&= q^{\frac{p}{4}-\frac{1}{6}} D(x, z; q^p) \prod_{n=0}^\infty \frac{1}{(1-z^2q^{n+1})(1-q^{n+1})^2(1-z^{-2}q^{n+1})}  \\
&=\sum_{s\in\Z}  \frac{x^sq^{-\frac{ps^2}{4}}}{\eta(q)} \sum_{m=0}^\infty q^{\frac{p}{4}-\frac{1}{8}}\frac{{\ch}[\rho_{2m+|s|+1}](z) q^{p\left(\frac{(2m+|s|+1)^2-1}{4}\right)}}{\prod\limits_{n=0}^\infty {(1-z^2q^{n+1})(1-q^{n+1})^2(1-z^{-2}q^{n+1})}  }\\
&=\sum_{s\in\Z}  \frac{x^sq^{-\frac{ps^2}{4}}}{\eta(q)} \sum_{m=0}^\infty {\ch}[L(2m+|s|)](z; q)
= \ch\left[ \cX_p \right](x, z; q)  
\end{split}
\end{equation}
\end{proof}

\subsubsection{On a rectangular W-algebra}\label{sec:rectangular}

We introduce the W-algebra $W_{\text{rect.}}^k(\mathfrak{sl}_{2n})$ obtained from the affine \voa{} of $\mathfrak{sl}_{2n}$ at level $k$ via quantum Hamiltonian reduction corresponding to the embedding of $\sltwo$ into  $\mathfrak{sl}_{2n}$ principally into both ``diagonal'' $\mathfrak{sl}_n$ sub algebras.  Its simple quotient is denoted by $W^{\text{rect.}}_k(\mathfrak{sl}_n)$. These W-algebras have been named rectangular \cite{ArM}; see that reference for more information on these W-algebras. 
\begin{cor}
The simple $W^{\text{rect.}}_{-5/2}(\mathfrak{sl}_4)$ is
\[
W^{\text{rect.}}_{-5/2}(\mathfrak{sl}_4) \cong \bigoplus_{m=0}^\infty L(2m)
\]
as $L_{-1}(\sltwo)$-module.
 \end{cor}
\begin{proof}
The Heisenberg coset in $\cS^2$ is $\text{Com}\left( \mathcal H, \cS^2 \right) \cong W^{\text{rect.}}_{-5/2}(\mathfrak{sl}_4)$ by Theorem 6.2 of \cite{CKLR}.
\end{proof}
\begin{cor}
The following extends $W^{\text{rect.}}_{-5/2}(\mathfrak{sl}_4)$ to a larger simple \voa
\[
\mathcal Y_1:=\bigoplus_{s\in2\Z}  \bigoplus_{m=0}^\infty L(2m+|s|)= \bigoplus_{m=0}^\infty (2m+1) L(2m)
\]
\end{cor}
\begin{proof}
This follows from Theorem 4.1 of \cite{CKLR} since
$\bigoplus\limits_{s\in2\Z} F_s$
is a lattice \voa.
\end{proof}
In fact $\mathcal Y_1$ is an infinite order simple current extension of $W^{\text{rect.}}_{-5/2}(\mathfrak{sl}_4)$ and thus serves as a nice example of the extensions envisagened in \cite{CKL}.
We will find nice generalizations of $W^{\text{rect.}}_{-5/2}(\mathfrak{sl}_4)$ and $\mathcal Y_1$ soon.

\subsection{Vertex operator algebras for $D_{2p}$ Argyres-Douglas theories}

For background on Quantum Hamiltonian reduction see \cite{Ar1, KW3}.
Let $W(n, k)$ denote the quantum Hamiltonian reduction of $V_k(\mathfrak{sl}_n)$ for the nilpotent element embedded principally in $\mathfrak{sl}_{n-2}$ and then the latter embedded in $\mathfrak{sl}_n$ such that
\[
\mathfrak{sl}_n \cong \mathfrak{sl}_{n-2} \oplus 2 \rho_{n-2} \oplus 2 \bar\rho_{n-2} \oplus 4  \mathbb C
\]
as $\mathfrak{sl}_{n-2}$-module. $\rho_{n-2}$ is the standard representation of $\mathfrak{sl}_{n-2}$ and $\bar\rho_{n-2}$ its conjugate.
The corresponding $\sltwo$-triplet $\{f, x, e\}$ embedded in $\mathfrak g=\mathfrak{sl}_n$ is described as follows. 
Identify $\mathfrak{sl}_n$ with the representation matrices in its standard representation and denote by $e_{i,j}$ the elementary matrices which have a zero everywhere except for the $(i, j)$-entry being one. Then the elements $\{f, x, e\}$ correspond to the matrices
\begin{equation}\label{eq:embsl2}
f= \sum_{i=1}^{n-3} e_{i+1, i}, \qquad x= \frac{1}{2}\sum_{i=1}^{n-2} (n-1-2i)e_{i, i}, \qquad e= \sum_{i=1}^{n-3} e_{i, i+1}.
\end{equation}
This W-algebra is strongly generated by fields of dimension $2, 3, \dots, n-2$ together with an affine $\mathfrak{gl}_2$ and dimension $(n-1)/2$ fields in the standard and conjugate representation of $\mathfrak{gl}_2$. 
The level of the $\mathfrak{sl}_2$ is $k+(n-3)$ and the central charge is
\[
c_{n, k} = \frac{k(n^2-1)}{k+n} -k(n-1)(n-2)(n-3) -(n-3)(n-4)(n^2-n-1).
\]
So that for $k=-n(n-2)/(n-1)$ we have that 
\[
k+n-3= -\frac{2n-3}{n-1} \qquad\text{and}\qquad c = 10-6n.
\] 
Define
\[
\cW_p := W\left(p+1, -\frac{p^2-1}{p} \right).
\]
Observe also that at this level $V_k(\mathfrak{sl}_2)$ has central charge
\[
c_p = 1-\frac{\frac{2n-3}{n-1}3}{-\frac{2n-3}{n-1}+2} =    10-6n = 4-6p
\]
so that it seems as we have a conformal embedding of $V_k(\mathfrak{sl}_2) \otimes \mathcal H$. Indeed for the first non-trivial cases, $p=2, 3$ this is true \cite{AKMP} and moreover the affine sub \voa{} is simple and acts completely reducible. Note that in general $V_k(\mathfrak{sl}_2)$ at the relevant levels, that is $k+2=\frac{1}{p}$ for positive integer $p$, is simple \cite{GK}. A proof that $V_k(\mathfrak{sl}_2) \otimes \mathcal H$ for $k+2=\frac{1}{p}$ embeds conformally in $\cW_p$ will be given in a moment.

 Recall the embedding of $\sltwo$ in $\mathfrak{sl}_{p+1}$ \eqref{eq:embsl2}.
We see that the subspace of $\mathfrak h$ that is annihilated by $f$, $\mathfrak h^f$, is spanned by 
\[
h_1^f:=e_{p,p}-e_{p+1, p+1}\quad\text{and}\quad h_2^f:=\frac{1}{p+1}\left((1-p)e_{p,p}+(1-p)e_{p+1, p+1}+2\sum\limits_{i=1}^{p-1} e_{i,i}\right).
\]
 We parameterize $z=vh_1^f+wh_2^f$ for complex $v, w$ and set $y=e(v)$ and $x=e(w)$. 
 In physics parlor $y$ is the fugacity associated to the $V_k(\mathfrak{sl}_2)$ subalgebra of $\cW_p$ and $x$ the one associated to the Heisenberg sub \voa.
 With this notation, we have
\begin{thm}
The character of $\cW_p$ satisfies
 \begin{equation}\nonumber
\begin{split}
\ch\left[ \cW_p \right] &= q^{\frac{p}{4}-\frac{1}{6}} \prod_{n=0}^\infty \prod_{a, b \in \{ \pm 1\}}\frac{\left(1-y^2q^{p(n+1)}\right)\left(1-q^{p(n+1)}\right)^2\left(1-y^{-2}q^{p(n+1)}\right)}{(1-y^2q^{n+1})(1-q^{n+1})^2(1-y^{-2}q^{n+1})\left(1-x^{a}y^{b}q^{p\left(n+\frac{1}{2}\right)}\right)}
\end{split}
\end{equation}
for $|q|^{-\frac{p}{2}} < |x|, |y| < |q|^{\frac{p}{2}}$.
\end{thm}
\begin{proof}
The character is stated explicitely in formula (11) of \cite{KW2}.
One has
\[
\ch\left[ \cW_p \right] = (-i)^{\frac{p(p+1)}{2}} q^{\frac{(p^2-1)(p-2)}{24} }  \frac{\eta(p\tau)^{\frac{3}{2}p -\frac{1}{2} (p^2+2p)} }{\eta(\tau)^{\frac{3}{2}p -\frac{1}{2} (p+6)}}  \frac{\prod\limits_{\alpha\in \Delta_+} \vartheta_{11}(p\tau, \alpha(z-\tau x))}{\prod\limits_{\alpha\in \Delta_+^0} \vartheta_{11}(\tau, \alpha(z)) \left( \prod\limits_{\beta \in  \Delta_{\frac{1}{2}}} \vartheta_{01}(\tau, \beta(z))\right)^{\frac{1}{2}}}
\]
$f, x, e$ is the $\sltwo$-triple associated to the $W$-algebra of $\mathfrak{sl}_{p+1}$. $\Delta_+$ denotes the positive roots and $\Delta_+^0$ the positive roots for which $\alpha(x)=0$. $\Delta_{\frac{1}{2}}$ are those positive roots for which $\beta(x)=1/2$. 
$z$ is an element of $\mathfrak h^f$, that is the $f$-invariant subspace of the Cartan subalgebra of $\mathfrak{sl}_{p+1}$ 
$\vartheta_{11}$ and $\vartheta_{10}$ are the standard Jacobi theta functions with product form
\[
\vartheta_{11}(\tau, z) = -i q^{\frac{1}{12}} u^{-\frac{1}{2}} \eta(\tau) \prod_{n=1}^\infty \left(1-u^{-1}q^n\right)\left(1-uq^{n-1}\right) 
 \]
 and
 \[
\vartheta_{01}(\tau, z) =\prod_{n=1}^\infty \left(1-u^{-1}q^{n-\frac{1}{2}}\right)\left(1-q^{n}\right) \left(1-uq^{n-\frac{1}{2}}\right).
 \]
 Here $u=e(z)$ and $q=e(\tau)$ as usual. 
 The character is viewed as a formal power series with the rule
$\frac{1}{1-x} = \sum\limits_{n=0}^\infty x^n$.
We call two formal power series $A, B$ equivalent, $A\sim B$, if they only differ by a factor of the form $\gamma q^ax^by^c$ for complex non-zero $\gamma$ and rational $a, b, c$, i.e. $A=\gamma q^ax^by^cB$. It is enough to show equivalence in this sense of the claim as both left-hand and right-hand side of the statement of the Theorem are formel power series of the form $q^{-\frac{c_p}{24}}(1 +\dots$. 

 Identify $\mathfrak{sl}_{p+1}$ with its standard representation as before. Let $\alpha_1, \dots, \alpha_p$ denote the $p$ simple positive roots of $\mathfrak{sl}_{p+1}$ such that they
 satisfy $\alpha_i(e_{i,i})=1, \alpha_i(e_{i+1,i+1})=-1$ and $\alpha_i(e_{j,j})=0$ otherwise. 
We thus see that $\alpha_i(z)=0$ for $i=1, \dots, p-2$ and 
$\alpha_{p-1}(z)=-v+w$ and $\alpha_p(z)=2v$. Moreover $\alpha_i(x)=1$ for $i=1, \dots, p-2$ and $\alpha_{p-1}(x)=1-\frac{p}{2}$ and $\alpha_p(x)=0$. 
The set of positive roots is 
\[
\Delta_+=\left\{ \left. \alpha_{i, j}:=\sum_{\ell=i}^j  \alpha_\ell \ \right|\ \ 1\leq i\leq j \leq p-1\right \}.
\]
We can now analyze each factor in the product of $\ch[\cW_p]$.
Firstly
\begin{equation}\nonumber
\begin{split}
\frac{\vartheta_{11}(p\tau, \alpha_{i,j}(z-\tau x))}{\eta(p\tau)} &\sim  \prod\limits_{n=1}^\infty \left(1-q^{p(n-1)+j+1-i}\right)\left(1-q^{pn+i-j-1}\right), \ \  j<p-1 \\
\frac{\vartheta_{11}(p\tau, \alpha_{i,p-1}(z-\tau x))}{\eta(p\tau)} &\sim \prod\limits_{n=1}^\infty\left(1-yx^{-1}q^{p\left(n-\frac{1}{2}\right)-i}\right)\left(1-xy^{-1}q^{p\left(n-\frac{1}{2}\right)+i}\right),\\
\frac{\vartheta_{11}(p\tau, \alpha_{i,p}(z-\tau x))}{\eta(p\tau)} &\sim \prod\limits_{n=1}^\infty\left(1-y^{-1}x^{-1}q^{p\left(n-\frac{1}{2}\right)-i}\right)\left(1-xyq^{p\left(n-\frac{1}{2}\right)+i}\right), \ \ i\neq p-1 \\
\frac{\vartheta_{11}(p\tau, \alpha_{p-1,p}(z-\tau x))}{\eta(p\tau)} &\sim \prod\limits_{n=1}^\infty\left(1-y^{-2}q^{pn}\right)\left(1-y^2q^{p(n-1)}\right).
\end{split}
\end{equation}
Secondly, $\Delta_0^+=\{\alpha_p\}$ if $p$ is odd and $\Delta_0^+=\{\alpha_p, \alpha_{\frac{p}{2}, p-1}, \alpha_{\frac{p}{2}, p}\}$ if $p$ is even.
Then
\begin{equation}\nonumber
\begin{split}
\vartheta_{11}(\tau, \alpha_{p}(z)) &\sim  \prod\limits_{n=1}^\infty\left(1-y^{-2}q^{n}\right)\left(1-y^2q^{(n-1)}\right) \left(1-q^{n}\right) \\
\vartheta_{11}(\tau, \alpha_{\frac{p}{2}, p-1}(z)) &\sim \prod\limits_{n=1}^\infty\left(1-yx^{-1}q^{n}\right)\left(1-xy^{-1}q^{n}\right)\left(1-q^{n}\right) \qquad \text{if} \ p \ \text{is even}\\
\vartheta_{11}(\tau, \alpha_{\frac{p}{2}, p}(z)) &\sim \prod\limits_{n=1}^\infty\left(1-y^{-1}x^{-1}q^{n}\right)\left(1-xyq^{n}\right)\left(1-q^{n}\right) \qquad \text{if} \ p \ \text{is even}
\end{split}
\end{equation}
Thirdly, $\Delta_{\frac{1}{2}}$ is empty if $p$ is even and $\Delta_{\frac{1}{2}}=\left\{\alpha_{\frac{p-1}{2}, p-1}, \alpha_{\frac{p-1}{2}, p},  -\alpha_{\frac{p+1}{2}, p-1}, -\alpha_{\frac{p+1}{2}, p}   \right\}$ if $p$ is odd. 
Then for odd $p$
\begin{equation}\nonumber
\begin{split}
\vartheta_{01}(\tau, \alpha_{\pm \frac{p\mp 1}{2}, p-1}(z)) &\sim \prod\limits_{n=1}^\infty\left(1-yx^{-1}q^{n-\frac{1}{2}}\right)\left(1-xy^{-1}q^{n-\frac{1}{2}}\right)\left(1-q^{n}\right) \\
\vartheta_{01}(\tau, \pm\alpha_{\frac{p\mp 1}{2}, p}(z)) &\sim \prod\limits_{n=1}^\infty\left(1-y^{-1}x^{-1}q^{n-\frac{1}{2}}\right)\left(1-xyq^{n-\frac{1}{2}}\right)\left(1-q^{n}\right).
\end{split}
\end{equation}
Plugging all these expressions into the formula for $\ch[\cW_p]$ and observing that 
\begin{equation}\nonumber
\begin{split}
\prod_{1\leq i \leq j <p-1} \frac{\vartheta_{11}(p\tau, \alpha_{i,j}(z-\tau x))}{\eta(p\tau)} &\sim
\prod\limits_{n=1}^\infty\prod_{1\leq i \leq j <p-1} \left(1-q^{p(n-1)+j+1-i}\right)\left(1-q^{pn+i-j-1}\right) \\
&\sim \prod\limits_{n=1}^\infty\prod_{1\leq i <p-1} \left(1-q^{pn+i-p}\right)^{p-1-i}\left(1-q^{pn-i}\right)^{p-1-i}\\
&\sim \prod\limits_{n=1}^\infty\prod_{1\leq i \leq p-1} \left(1-q^{pn-i}\right)^{p-2} \\
&\sim \prod\limits_{n=1}^\infty \frac{(1-q^n)^{p-2}}{(1-q^{pn})^{p-2}} \sim \frac{\eta(\tau)^{p-2}}{\eta(p\tau)^{p-2}}
\end{split}
\end{equation}
as well as (for $a, b \in \{ \pm 1\}$) 
\[
\prod\limits_{n=1}^\infty \prod_{1\leq i \leq p-1} \frac{\left(1-y^ax^bq^{p\left(n-\frac{1}{2}\right)-i}\right)}{\left(1-y^ax^bq^{\left(n-\frac{p}{2}\right)}\right)}\sim \prod\limits_{n=1}^\infty \frac{1}{\left(1-y^ax^bq^{\left(pn-\frac{p}{2}\right)}\right)}
\]
gives the result.
\end{proof}
\begin{rem}
The $\mathcal B_p$-algebra is conjecturally also a certain subregular Quantum Hamiltonian reduction  of $V_k(\mathfrak{sl}_{p-1})$ at level $k+p-1=\frac{p-1}{p}$. The character of that reduction can be messaged similarly as in the above theorem and in that way one can at least verify that conjecture on the level of characters. This analysis will be carried out as part of the general analysis of the representation category of the $\cB_p$-algebras \cite{ACKR}.
\end{rem}
Recall that $\cX_p =  \bigoplus\limits_{m=0}^\infty  L(m) \otimes \left( F_{-\lambda m} \oplus F_{-\lambda(m-2)} \oplus \dots \oplus F_{\lambda (m-2)} \oplus F_{\lambda m}\right)$
with $\lambda^2=-\frac{p}{2}$ as a module for $L_k(\sltwo)\otimes \mathcal H$,.
The previous Theorem together with Proposition \ref{prop:main} tell us 
\begin{thm}\label{thm:main}
The character of $\cW_p$ satisfies $\ch[\cW_p](x, y; \tau) = \ch[\cX_p](x, y; \tau)$. 
\end{thm}
Recall that the modules $L(m)$ are simple and that the corresponding category $\cO$ is semi-simple. The characters of the $L(m)$ are moreover clearly linearly independent as meromorphic functions in $y$ and $q$. 
The above statement thus also holds on the level of modules:
\begin{cor}\label{cor:confemb}
$L_k(\sltwo)\otimes \mathcal H$ for $k+2=\frac{1}{p}$ embeds conformally in $\cW_p$
and  as $L_k(\sltwo)\otimes \mathcal H$-modules $\cW_p\cong \cX_p$. 
\end{cor}

For the very same reason as in Section \ref{sec:rectangular}, that is Theorem 4.1 of \cite{CKLR}, we have 
\begin{cor}\label{cor:QH} The following two $L_k(\sltwo)$-modules inherit the structure of a simple vertex operator algebra from $\cW_p$
\[
\mathcal C_p := \text{Com}\left( \mathcal H, \cW_p \right) \cong \bigoplus_{m=0}^\infty  L(2m)
\qquad
\text{and}
\qquad
\mathcal Y_p :=  \bigoplus_{m=0}^\infty (2m+1) L(2m).
\]
\end{cor}
It is interesting to note the following
\begin{prop}\label{prop:red}
As Virasoro \voa-modules, we have
\[
H_{+}^0(\cW_p ) \cong \mathcal B_{p}, \qquad H_{+}^0(\mathcal C_p ) \cong \mathcal M(p) \qquad \text{and} \qquad H_{+}^0(\mathcal Y_p ) \cong \mathcal W(p).
\]
Here $H_+$ is the $+$ reduction functor of the DS-reduction from $\mathfrak{sl}_2$ to the Virasoro \voa, see \cite{Ar2}.
\end{prop}
\begin{proof}
We have $H_{+}^0(L(k,m) ) \cong  L(h_{m+1, 1}, c_p)$ as a direct consequence of \cite[Theorem 9.1.4]{Ar2} with $c_p=1-6(p-1)^2/p$ so that the claim follows from the decomposition of $\mathcal B_p$, 
$\Sing$ and $\Trip$ in equations \ref{eq:Bpdecomp}, \ref{eq:Mpdecomp} and \ref{eq:Wpdecomp}.
\end{proof}
\begin{conj}
As vertex operator algebras
\[
H_{+}^0(\cW_p ) \cong \mathcal B_{p}, \qquad H_{+}^0(\mathcal C_p ) \cong \mathcal M(p) \qquad \text{and} \qquad H_{+}^0(\mathcal Y_p ) \cong \mathcal W(p).
\]
\end{conj}
\begin{rem}
Drazen Adamovi\'c announced the study of certain \voa s $R(p)$ \cite{A2}. I believe that they coincide with $\cW_p$. For $p=3$ this is true \cite{AKMP} but also had to be true due to the uniqueness Theorem of minimal W-algebras, see Theorem 3.1 of \cite{ACKL}. 
Understanding in general the relation between $R(p)$ and $\cW_p$ might lead to a proof of the above conjecture. 
\end{rem}

\end{document}